\documentclass[10pt,conference]{IEEEtran}
\usepackage{graphics,graphicx,psfrag,color,float, hyperref}
\usepackage{epsfig}
\usepackage{amsmath}
\usepackage{amssymb}
\usepackage{amsfonts}
\usepackage {times}
\usepackage{braket}
\usepackage{amsfonts}
\usepackage[T1]{fontenc}  
\usepackage{enumerate}

\newcommand{\beq}{\begin{equation}}
\newcommand{\enq}{\end{equation}}
\newcommand{\bel}{\begin{lemma}}
\newcommand{\enl}{\end{lemma}}
\newcommand{\bet}{\begin{theorem}}
\newcommand{\ent}{\end{theorem}}

\newcommand{\eps}{\varepsilon}

\newcommand*{\cA}{\mathcal{A}}

\newcommand*{\cX}{\mathcal{X}}
\newcommand*{\cB}{\mathcal{B}}
\newcommand*{\cY}{\mathcal{Y}}

\newcommand*{\cM}{\mathcal{M}}

\newcommand*{\bX}{\mathbf{X}}

\newcommand*{\bY}{\mathbf{Y}}
\newcommand*{\by}{\mathbf{y}}

\newcommand*{\bx}{\mathbf{x}}

\newcommand*{\bP}{\mathbf{P}}
\newcommand*{\bQ}{\mathbf{Q}}

\newcommand*{\Ib}{\bar{I}}
\newcommand*{\hX}{\hat{X}}

\newcommand*{\oH}{\overline{H}}

\mathchardef\mhyphen="2D

\newcommand*{\renyi}{R\'{e}nyi }
\newcommand*{\verdu}{V\'{e}rdu }

\newcommand*{\denc}{\mathrm{d-enc}}
\newtheorem{definition}{Definition}
\newtheorem{theorem}{Theorem}
\newtheorem{lemma}{Lemma}

\begin{document}
\title{One-shot source coding with coded side information available at the decoder}
\author{
\authorblockN{Naqueeb~Ahmad~Warsi}
\authorblockA{Tata Institute of Fundamental Research\\
Homi Bhabha Road, Mumbai 400005\\
Email: naqueeb@tifr.res.in}}
\maketitle
\begin{abstract}
One-shot achievable rate region for source coding when coded side information is available at the decoder (source coding with a helper) is proposed. The achievable region proposed is in terms of  conditional smooth max \renyi entropy and smooth max \renyi divergence. Asymptotically (in the limit of large block lengths) this region is quantified in terms of spectral-sup conditional entropy rate and spectral-sup mutual information rate. In particular, it coincides with the rate region derived in the limit of unlimited independent and identically distributed copies of the sources.
\end{abstract}

\section{Introduction}
The derivation of most of the fundamental results in information theory relies on the assumption that a random experiment is repeated identically and independently for large enough time. However, in practical scenarios both of these assumptions are not always justifiable. To overcome the limitations posed by these assumptions Renner et al. introduced the notion of \emph{one-shot} information theory. One-shot information theory relies on the fact that a random experiment is available only \emph{once}. Thus removing both the assumptions together.

The first one-shot bounds were given for the task of one-shot source coding \cite{renner-wolf-2004}. These bounds were based on smooth \renyi entropies. The notion of smooth \renyi entropies were introduced for the very first time in the same work, i.e., in Ref. \cite{renner-wolf-2004}. The elegance of the one-shot bounds obtained in Ref. \cite{renner-wolf-2004} is that these bounds coincide with the Shannon entropy \cite{covertom} of the information source in the limit of unlimited independent and identically distributed (i.i.d.) copies of the source. Furthermore, these bounds coincide with spectral sup-entropy rate as defined by Han and \verdu in Ref. \cite{han-verdu-1993} in the limit of unlimited arbitrarily distributed copies of the source. One-shot bounds for distributed source coding were given by Sharma et al. in \cite{sharma-warsi-2011}. In \cite{renner-isit-2009} Wang et al. gave one-shot bounds for the channel coding problem in terms of smooth min \renyi divergence. 

There has been a considerable work on the one-shot bounds for the quantum case under various scenarios (see for
example Refs. \cite{datta-renner-2009, konig-op-mean-2009,dupuis-broadcast-2010,berta-reverse-shannon-2011,
datta-fqsw-2011,renes-renner-2011} and references therein).

In this work we give one-shot achievable rate region for source coding  when coded state side information is available at the decoder. The achievable rate region derived for this problem is in terms of smooth max \renyi divergence and conditional smooth max \renyi entropy. The notion of smooth max \renyi divergence was introduced by Datta for the quantum case in \cite{Datta-2009}. We further show that the achievable region obtained asymptotically coincides with the rate region derived in \cite{wyner-1975}.

The rest of this paper is organized as follows. In Section \ref{notation} we  
discuss the notations which we will be using throughout this paper. In Section \ref{div} we give the definitions of of smooth conditional \renyi entropy of order zero and smooth max \renyi divergence. We then prove two lemmas pertaining to the properties of smooth max \renyi divergence. Although, the proof of Lemma \ref {asyminfinityrate} is known in the quantum case we give a totally different proof. In particular, our proof involves more straight forward arguments. In Section \ref{sideinf} we state and prove the achievable region for source coding problem when coded side information is available at the decoder. 
\section{Notations}
\label{notation}
In the discussions below we will be using $X$ to represent a random variable. We will assume that all the random variables are discrete and have finite range. We represent a random sequence of length $n$ by $X^n$ and a particular realization of $X^n$ by $\bx$. Notation $\bX$ will be used to represent an arbitrary sequence of random variables, i.e., $\bX = \{X_n\}_{n=1}^{\infty}$. We use the notation $|\cdot|$ to represent the cardinality of a set. The set $\{\bx : P_{X^n}(\bx)>0\}$ is denoted by $\mbox{Supp} (P_{X^n})$. We use the notation
\beq
X\rightarrow Y \rightarrow Z \nonumber
\enq
to denote the fact that random variables $X$, $Y$ and $Z$ form a Markov chain. We represent the following set of real numbers
\beq
\{x : 0\leq x < \infty\} \nonumber
\enq
by $\mathbb{R}^+$. $\cX \times \cY$ will represent the cartesian product of two sets. Similarly $(\cX \times \cY)^n$ will represent the $n\mhyphen$th Cartesian product of the set $\cX \times \cY$. The notation $\mathbb{N}$ is used to represent the set of natural numbers. Throughout this paper we will assume that $\log$ is to the base $2$.
\section{Smooth \renyi divergence of order infinity and conditional smooth \renyi entropy of order zero}
\label{div}

\begin{definition}(Max \renyi entropy \cite{renyi-1960})
Let $X \sim P_{X}$, with range $\cX$. The zero order \renyi entropy of $X$ is defined as
\beq
H_{0}(X) := \log \mbox{Supp}(P_X). \nonumber
\enq
\end{definition}
\begin{definition}(Conditional smooth max \renyi entropy \cite{renner-wolf-2005})
\label{cond}
Let $(X,Y) \sim P_{XY}$, with range $\cX \times \cY$. For $\eps \geq 0$, the conditional smooth \renyi entropy of order zero of $X$ given $Y$ is defined as
\begin{align*}
H_{0}^{\eps} (X|Y) := \min_{Q \in \cB^{\eps}(P_{XY})} \log\max_{y\in\cY} |\mbox{Supp}(Q(X|Y =y))|, \nonumber
\end{align*}
where $\cB^{\eps}(P_{XY}) = \{Q: \sum_{x,y \in \cX \times \cY} Q(x,y) \geq 1 -\eps, \forall ~(x,y) \in \cX \times \cY,~ P_{XY}(x,y)\geq Q(x,y) \geq 0\}$ and $Q(X=x|Y=y):=\frac{Q(x,y)}{P_{Y}(y)}$, for any $x \in \cX$ and $y \in \cY$. With the convention that $Q(X=x|Y=y) := 0$ if $P_{Y}(y) = 0$.
\end{definition}

\begin{definition}(Max \renyi divergence \cite{renyi-1960})
Let $P$ and $Q$ be two probability mass functions on the set $\cX$ such that $\mbox{Supp}(P) \subseteq \mbox{Supp}(Q)$. The max \renyi divergence between $P$ and $Q$ is defined as 
\beq
D_{\infty}(P||Q) :=  \log \max_{x : P(x)>0} \frac{P(x)}{Q(x)}. \nonumber
\enq
\end{definition}
\begin{definition}(Smooth max \renyi divergence)
\label{smoothorderinf} Let $P$ and $Q$ be two probability mass functions on the set $\cX$ such that $\mbox{Supp}(P) \subseteq \mbox{Supp}(Q)$.  The smooth max \renyi divergence between $P$ and $Q$ for $\eps \in [0,1)$ is defined as
\beq
D^{\eps}_{\infty}(P||Q) := \log \inf_{\phi \in \cB^{\eps}(P)} \max_{x : P(x)>0} \frac{\phi(x)}{Q(x)}, \nonumber
\enq
where
\begin{align*}
 \cB^{\eps}(P) = \bigg\{&\phi : 0\leq \phi(x) \leq P(x),\forall x \in \cX~\mbox{and}\\
 &\sum_{x\in \cX} \phi(x) \geq 1-\eps\bigg\}.
 \end{align*}
 Notice that $D^{\eps}_{\infty}(P||Q)$ is a non increasing function of $\eps$. 
\end{definition}
\begin{lemma} (Datta and Renner \cite{datta-renner-2009})
\label{assymcond}
Let $(\bX, \bY) = \{(X_n,Y_n)\}_{n=1}^{\infty}$ be an arbitrary random  sequence taking values over the set $\{(\cX\times\cY)^n\}_{n=1}^{\infty}$,
where $(\cX\times\cY)^n$ is the $n\mhyphen$th Cartesian product of $\cX\times \cY$. Then 
\beq
\lim_{\eps\to 0}\limsup_{n\to\infty} \frac{H_{0}^\eps(X^n|Y^n)}{n} = \oH(\bX|\bY), \nonumber
\enq
where 
\begin{align}
&\oH(\bX|\bY) :=\nonumber\\
& \mbox{inf}\bigg\{\alpha\big|\liminf_{n\to\infty}\Pr\bigg\{\frac{1}{n} \log \frac{1}{P_{X^n|Y^n}(X^n|Y^n)} \leq \alpha\bigg\} = 1 \bigg\}. \nonumber
\end{align}
$\oH(\bX|\bY)$ is called the spectral-sup conditional entropy rate of $\bX$ given $\bY$ \cite{han}. In particular, if $(\bX, \bY) = \{(X_n,Y_n)\}_{n=1}^{\infty}$ is a random sequence of independent and identically distributed random pairs distributed according to $P_{XY}$ then
\beq
\lim_{\eps\to 0}\limsup_{n\to\infty} \frac{H_{0}^\eps(X^n|Y^n)}{n} = H(X|Y).\nonumber
\enq
\end{lemma}
\begin{lemma}
\label{consopt}
Let $P$ and $Q$ be two probability mass functions defined on the set $\cX$, where $\mbox{Supp}(P) \subseteq \mbox{Supp}(Q)$ and $|\cX|< \infty$. There exists $\phi \in \cB^{\eps}(P)$ such that 
\beq
\label{maxinfsm}
D^{\eps}_{\infty}(P||Q) = \log  \max_{x \in \cX} \frac{\phi(x)}{Q(x)}. 
\enq 
\end{lemma}
\begin{proof}
Without loss of generality we assume here that $\cX \subset \mathbb{N}$.
We construct the  $\phi_{i}$s by decreasing the $P_i$s such that the total decrease is $\eps$, i.e., $\sum_{i \in \cX}(P_{i}-\phi_i )= \eps$. The following procedure achieves the above.
\begin{enumerate}[Step 0 $\rightarrow$]
\item (Initialization) $\phi_i = P_i~ \forall i \in \cX $.
\end{enumerate}
\begin{enumerate}[Step 1 $\rightarrow$]
\item Let $r_1, r_2$ be the largest and the second largest ratios in 
\beq
\cA = \bigg\{\frac{\phi_i}{q_i} : i \in \cX\bigg\}
\enq
respectively. Let I be the collection of all $i$s that have the highest ratio, i.e., 
 \beq
I = \bigg\{ i \in \cX: \frac{\phi_i}{Q_i} = r_1\bigg\}.
\enq
If $|\cA| =1$ then notice that $I = \cX$. In this case start decreasing all $\phi_i$s where $i \in \cX$ such that all the indices continue to have constant ratio, i.e., $\frac{\phi_i}{Q_i} = \frac{\phi_j}{q_j} ~\forall i,j \in \cX$. Continue this process until we run out of $\eps$, i.e., $\sum_{i \in \cX}(P_i- \phi_i) = \eps$ in which case end the procedure. Else go to step $2$.
\end{enumerate}
\begin{enumerate}[Step 2 $\rightarrow$]
\item We start decreasing all $\phi_i$s where $i \in I$ such that indices in $I$ continue to have the highest ratio, i.e., $\frac{\phi_i}{Q_i} = \frac{\phi_j}{q_j} ~\forall i,j \in I$. As a result, $r_1$ will start decreasing. Continue decreasing till either 
\begin{enumerate}
\item [Case $1$:] $r_1$ hits $r_2$, i.e., $r_1 = r_2$ in which case stop decreasing any further. Goto step $1$. Or \\
\item [Case $2$:] we run out of $\eps$, i.e., $\sum_{i \in I}(P_i- \phi_i) = \eps$ in which case end the procedure.\\
\end{enumerate}
\end{enumerate}
We claim that the $\phi$ constructed by the above procedure is such that
\beq
\label{c1}
\log  \max_{x \in \cX} \frac{\phi(x)}{Q(x)} = D^{\eps}_{\infty}(P||Q).
\enq 
We give a proof by contradiction to prove \eqref{c1}. Let $\phi^\prime \in \cB^{\eps}(P)$ be the output of some other procedure such that
\beq
\label{contrad1222}
\log  \max_{x \in \cX} \frac{\phi(x)}{Q(x)}  > \log  \max_{x \in \cX} \frac{\phi^\prime(x)}{Q(x)}. 
\enq
Let $\hat{\cA}= \big\{i\in \cX : \phi_i < P_i \big\}$. Notice that for every $i,j \in \hat{\cA}$
\beq
\frac{\phi_i}{Q_i} = \frac{\phi_j}{Q_j}. \nonumber
\enq
It is easy to observe that for \eqref{contrad1222} to hold $\phi^\prime$ must satisfy the following
\beq
\label{np}
\phi^\prime_i<\phi_i ,~\forall i \in \cX.
\enq
However, this is not possible because this new procedure will not have enough $\eps$ to accomplish \eqref{np}, i.e.,
\beq
\sum_{i \in \hat{\cA}}(P_i- \phi^\prime_i) > \eps. \nonumber
\enq 
\end{proof}
{\bf{Remark:}} It is easy to observe from the proof of Lemma \ref{consopt} that for $\eps \in [0,1)$,
\beq
\label{equalsupp}
\mbox{Supp}(\phi) = \mbox{Supp}(P),
\enq
where $\phi$ satisfies \eqref{maxinfsm}.

\begin{lemma}
\label{asyminfinityrate}
Let $\bP = \{P_n\}_{n=1}^{\infty}$ and $\bQ = \{Q_n\}_{n=1}^{\infty}$ be an arbitrary sequences of probability mass functions defined on the set $\{\cX^n\}_{n=1}^{\infty}$, where $\cX^n$ is the $n\mhyphen$th cartesian product of the set $\cX$ and $|\cX| < \infty$. Assume that for every $n \in \mathbb{N}$, $\mbox{Supp}(P_n) \subseteq \mbox{Supp}(Q_n)$. Then
\beq
\label{noniid}
\lim_{\eps \to 0} \limsup_{n \to \infty} \frac{1}{n}D^{\eps}_{\infty}(P_n||Q_n) = \Ib(\bP;\bQ),
\enq
where 
\begin{align}
\label{spectralsupmutual}
\Ib&(\bP;\bQ):= \mbox{inf}\left\{\alpha\big|\liminf_{n\to\infty}\Pr\left\{\frac{1}{n}\log \frac{P_n}{Q_n} \leq \alpha\right\} = 1 \right\}.
\end{align}
$\Ib(\bP;\bQ)$ is called the spectral sup-mutual information rate between $\bP$ and $\bQ$ \cite{han}. In particular, if $\bP = \{P^{\times n}\}_{n =1}^{\infty}$ and $\bQ = \{Q^{\times n}\}_{n =1}^{\infty}$, where $P^{\times n}$ and $Q^{\times n}$ represent the product distributions of $P$ and $Q$ on $\cX^n$. Then
\beq
\label{iid}
\lim_{\eps \to 0} \limsup_{n \to \infty} \frac{1}{n}D^{\eps}_{\infty}(P_n||Q_n)=D(P||Q).
\enq
\end{lemma}
\begin{proof}
We will first prove 
\beq
\lim_{\eps \to 0} \limsup_{n \to \infty} \frac{1}{n}D^{\eps}_{\infty}(P_n||Q_n) \leq \Ib(\bP;\bQ). \nonumber
\enq
Consider any $\lambda > \Ib(\bP;\bQ)$. Let us define the following set 
\beq
\label{consset}
\cA_{n}(\lambda) := \left\{\bx : \frac{1}{n} \log \frac{P_n(\bx)}{Q_n(\bx)} \leq \lambda\right\}. 
\enq
Let $\phi_{n}: \cX^n \to [0,1]$, $n \in \mathbb{N}$, such that 
\beq
\label{constsupinfrate}
\phi_n(\bx) =
\begin{cases}
P_{n}(\bx) & \mbox{if } \bx \in \cA_{n}(\lambda), \\
 0      & \mbox{otherwise.}
\end{cases}
\enq
From  \eqref{spectralsupmutual} it easily follows that 
\beq 
\lim_{n \to \infty} \Pr\{\cA_{n}(\lambda)\} =1.
\enq
Thus from our  construction of $\phi_n$, \eqref{constsupinfrate}, it follows that
\beq
\label{constsminf}
\lim_{n \to \infty}\sum_{\bx \in \cX^n} \phi_n(\bx) = \lim_{n \to \infty} \Pr\{\cA_{n}(\lambda)\} =1.
\enq 
Using \eqref{constsupinfrate} and \eqref{constsminf} observe that for $n$ large enough 
\begin{align*}
\lim_{\eps \to 0} \limsup_{n \to \infty}\frac{1}{n}D^{\eps}_{\infty}(P_n||Q_n) & \leq \limsup_{n \to \infty} \frac{1}{n} \log \max_{\bx \in \cA_{n}(\lambda)} \frac{\phi_n(\bx)}{Q_{n}(\bx)}\\
&\overset{a} \leq \lambda,
\end{align*}
where $a$ follows from \eqref{consset} and \eqref{constsupinfrate}.

We now prove the other side, i.e., 
\beq
\lim_{\eps \to 0} \limsup_{n \to \infty} \frac{1}{n}D^{\eps}_{\infty}(P_n||Q_n) \geq \Ib(\bP;\bQ). \nonumber
\enq
Consider any $\gamma < \Ib(\bP;\bQ)$. For every $n \in \mathbb{N}$, let us define the following the set 
\beq
\label{conssetasym0}
\cA_{n}(\gamma) := \left\{\bx: \frac{1}{n} \log\frac{P_n(\bx)}{Q_{n}(\bx)} \geq \gamma \right\}.
\enq
From  \eqref{spectralsupmutual} it follows that there exists $\eta \in (0,1]$, such that 
\beq
\label{conssetasym}
\limsup_{n\to\infty}\Pr\{\cA_{n}(\gamma)\} = \eta. 
\enq
 Since $\Pr\{\cA_{n}(\gamma)\}+\Pr\{\cA^c_{n}(\gamma)\} =1$, for every $n \in \mathbb{N}$, we have 
 \beq
 \label{imp000}
 \liminf_{n\to\infty}\Pr\{\cA^c_{n}(\gamma)\} = 1-\eta.
 \enq
For every $\eps \in (0, \eta)$, let us define a sequence of positive functions $\{\phi_{n}\}_{n =1}^\infty$, such that for every $n \in \mathbb{N}$
\begin{align}
&\phi_{n}:\cX^n \to [0,1], \phi_{n}(\bx) \leq P_n(\bx), \forall \bx \in \cX^n \nonumber\\
\label{conssetasym1}
&\mbox{and} ~ \sum_{\bx \in \cX^n} \phi_{n}(\bx) \geq 1-\eps.
\end{align}
We now claim that for large enough $n$, $\mbox{Supp}(\phi_n) \cap \cA_{n}(\gamma) \neq \phi$. To prove this claim, suppose that $\mbox{Supp}(\phi_n) \cap \cA_{n}(\gamma) = \phi$. This would further imply that 
\begin{align}
\liminf_{n\to\infty}\sum_{\bx\in \cX^n}\phi_{n}(\bx)& \leq \liminf_{n\to\infty}\Pr\{\cA^c_{n}(\gamma)\} \nonumber\\
& \overset{a}= 1-\eta \nonumber\\
\label{contradiction}
& \overset{b}< 1-\eps,
\end{align}
where $a$ follows from \eqref{imp000} and $b$ follows because $\eps < \eta$. Notice that \eqref{contradiction} contradicts \eqref{conssetasym1}. 

 Thus for $n$ large enough,
\begin{align*}
 1-\eps & \leq\sum_{\bx \in \cA^c_{n}(\gamma)} \phi_{n}(\bx)+\sum_{\bx \in \cA_{n}(\gamma)}\phi_{n}(\bx) \nonumber\\
 & \leq \Pr\{\cA^c_{n}(\gamma)\} + \sum_{\bx \in \cA_{n}(\gamma)}\phi_{n}(\bx).
\end{align*}
By rearranging the terms in the above equation we get
\beq
\label{imp00}
1-\eps-\Pr\{\cA^c_{n}(\gamma)\} \leq \sum_{\bx \in \cA_{n}(\gamma)}\phi_{n}(\bx).
\enq
Taking $\limsup$ on both sides of \eqref{imp00}, we have
\begin{align}
\limsup_{n \to \infty}\sum_{\bx \in \cA_{n}(\gamma)}\phi_{n}(\bx)
 &\geq 1 -\eps - \liminf_{n\to\infty}\Pr\{\cA^c_{n}(\gamma)\}\nonumber\\
 \label{imp001}
 &\geq \eta - \eps.
\end{align}
\eqref{imp001} follows from \eqref{imp000}.
Now notice the following set of inequalities for large enough $n$
\begin{align}
1&\geq \sum_{\bx \in \cA_{n}(\gamma)}P_{n}(\bx)\nonumber\\ 
& \overset{a} \geq \sum_{\bx \in \cA_{n}(\gamma)}2^{n\gamma}Q_{n}(\bx)\nonumber\\ 
\label{rearrange1}
& \overset{b} \geq 2^{\left(n\gamma-\max_{\bx\in\cX^n}\log\frac{\phi_n(\bx)}{Q_{n}(\bx)}\right)}\sum_{\bx\in\cA_n(\gamma)}\phi_{n}(\bx)
\end{align}
where $a$ follows from \eqref{conssetasym0}; b follows from the fact that for every $\bx \in \cA^n(\gamma)$,
\beq
 \frac{\phi_n(\bx)}{Q_{n}(\bx)} \leq \max_{\bx\in\cA^n(\gamma)}\frac{\phi_n(\bx)}{Q_{n}(\bx)} \leq \max_{\bx\in\cX^n}\frac{\phi_n(\bx)}{Q_{n}(\bx)} \nonumber
 \enq
 By taking $\log$ on both sides of \eqref{rearrange1} and rearranging the terms we get 
 \beq
 \max_{\bx\in\cX^n}\frac{1}{n}\log\frac{\phi_n(\bx)}{Q_{n}(\bx)}\geq \gamma+\frac{1}{n}\log \sum_{\bx\in\cA_{n}(\gamma)}\phi_{n}(\bx). \nonumber\\
 \enq
 Taking $\limsup$ on both sides of the above equation we have
 \begin{align}
\limsup_{n \to \infty} \max_{\bx\in\cX^n}\frac{1}{n}\log\frac{\phi_n(\bx)}{Q_{n}(\bx)} & \geq \gamma+ \limsup_{n\to \infty}\frac{1}{n}\log \sum_{\bx\in\cA_{n}(\gamma)}\phi_{n}(\bx) \nonumber\\
\label{finsup}
& \geq \gamma
 \end{align}
 where \eqref{finsup} follows from \eqref{imp001}. Notice that \eqref{finsup} is true for every $\phi_n$ satisfying \eqref{conssetasym1}. Thus
\beq
\label{fin123}
\limsup_{n \to \infty} \frac{1}{n}D^{\eps}_{\infty}(P_n||Q_n) \geq \gamma.
\enq
Since \eqref{fin123} is true for every $\eps \in (0,\eta)$, the result will hold true for $\eps \downarrow 0$. 

\eqref{iid} easily follows from the law of large numbers and \eqref{noniid}. This completes the proof.
\end{proof}

\section{Source coding with coded state side information available at the decoder}
Let $(X^n, Y^n) \sim P_{X^nY^n}$, with range $(\cX\times\cY)^n$, where 
\beq
(X^n , Y^n) := [(X_1, Y_1), (X_2, Y_2),\dots, (X_n, Y_n)]. \nonumber
\enq
The $n$-shot source coding with coded side information available at the decoder is formulated as follows. We first define two sets of integers
\begin{align}
\label{c1}
\cM_n^{(1)} & = \{1,\dots, 2^{\ell_{\denc}^{\eps}(X^n)}\},\\
\label{c2}
\cM_n^{(2)} & = \{1,\dots, 2^{\ell_{\denc}^{\eps}(Y^n)}\}
\end{align} 
called the codes. Choose arbitrary mappings $e_n^{(1)}: \cX^n \to \cM_n^{(1)}$ (encoder$1$) and $e_n^{(2)}: \cY^n \to \cM_n^{(2)}$ (encoder$2$). We call
\begin{align*}
\frac{\ell_{\denc}^{\eps}(X^n)}{n} & = \frac{\log |\cM_n^{(1)}|}{n},\\
\frac{\ell_{\denc}^{\eps}(Y^n)}{n} & = \frac{\log |\cM_n^{(2)}|}{n}
\end{align*}
the coding rates of the encoder $1$ and encoder $2$, respectively. The decoder $d_n : \cM_n^{(1)} \times \cM_n^{(2)} \to \cX^n$ receives two outputs $e_n^{(1)}(\bx)$ and $e_n^{(2)}(\by)$ from the two encoders and tries to reconstruct the original source output $\bx$. Thus the probability of error for this task is defined as
\beq
P^n_{e} := \Pr\{X^n \neq \hX^n\}, \nonumber
\enq
where $\hX^n = d_n(e_n^{(1)}(X^n),e_n^{(2)}(Y^n))$. Note here that the encoders $e_n^{(1)}$ and $e_n^{(2)}$ do not cooperate with each other. We call the triplet $(e_n^{(1)},e_n^{(2)}, d_n)$ of two encoders and one decoder with the two codes in \eqref{c1} and \eqref{c2} and the error probability $\eps$ the $(n, 2^{\ell_{\denc}^{\eps}(X^n)},2^{\ell_{\denc}^{\eps}(Y^n)}, \eps)$ $n$-shot code. 

In this coding system we wish to minimize the two coding rates $\frac{\ell_{\denc}^{\eps}(X^n)}{n}$ and $\frac{\ell_{\denc}^{\eps}(Y^n)}{n}$ such that the probability of error is less than $\eps$.
\begin{definition}(One-shot  $\eps$ achievable rate pair) 
\label{oneshotrate} 
Let $(X,Y)\sim P_{XY}$, with range $\cX \times \cY$. A one-shot rate pair $(R_1,R_2)$ is called $\eps$ achievable if and only if there exists a $(1, 2^{\ell_{\denc}^{\eps}(X)},2^{\ell_{\denc}^{\eps}(Y)}, \eps)$ one-shot code such that  $\Pr\{X \neq \hX\} \leq \eps$, $\ell_{\denc}^{\eps}(X)\leq R_1$ and $\ell_{\denc}^{\eps}(Y) \leq R_2$.
\end{definition}

\begin{definition}(Asymptotically achievable rate pair)
\label{asymsiinf}
A rate pair $(R_1, R_2)$ is asymptotically achievable if and only if there exists $(n, 2^{\ell_{\denc}^{\eps}(X^n)},2^{\ell_{\denc}^{\eps}(Y^n)}, \eps)$ code such that  $\Pr\{X^n \neq \hat{X}^n\} \leq \eps$,
\beq
\lim_{\eps \to 0} \limsup_{n \to \infty}\frac{\ell_{\denc}^{\eps}(X^n)}{n}\leq R_1 \nonumber
 \enq
  and
 \beq
 \lim_{\eps\to 0}\limsup_{n \to \infty}\frac{\ell_{\denc}^{\eps}(Y^n)}{n}\leq R_2. \nonumber
\enq
\end{definition}

\label{sideinf}
\begin{theorem}
\label{oneshotsideinf}
Let $(X,Y) \sim P_{XY}$, with range $\cX\times\cY$. For the error $\eps\in(0,1)$. The following one-shot rate region for source coding of $X$ with a helper observing $Y$ is achievable 
\begin{align*}
\ell_{\denc}^{\eps}(X) & \geq H^{\eps_{11}}_{0}(X|U) - \log (\eps -\eps_1),\\
\label{swr2}
\ell_{\denc}^{\eps}(Y) & \geq D^{\eps_{11}}_{\infty}(P_{UY}||P_U\times P_Y)\\
&\hspace{5mm}+\log[-\ln (\eps_1-\eps_{11}-2\eps_{11}^{\frac{1}{2}}) ]
\end{align*}
for some conditional pmf $P_{U|Y}$, where $\eps_1 < \eps$ and $\eps_{11}$ is such that 
\beq
\label{condition123}
\eps_{11}+ 2\eps_{11}^{\frac{1}{2}} < \eps_1~\mbox{and}~D^{\eps_{11}}_{\infty}(P_{UY}||P_U\times P_Y) \geq 0.
\enq
\end{theorem}\begin{proof}
The techniques used in the proof here are motivated from \cite[Lemma 4.3]{{wyner-1975}}. Fix a conditional probability mass function $P_{U|Y}$ and let $P_{U}(u) = \sum_{y\in \cY}P_{Y}(y)P_{U|Y}(u|y)$. Choose $\eps_{11}$ such that the conditions in \eqref{condition123} are satisfied. Notice that such a choice of $\eps_{11}$ always exists because $D^{\eps_{11}}_{\infty}(P_{UY}||P_U\times P_Y)$ is a decreasing function of $\eps_{11}$.
Let $Q \in \cB^{\eps_{11}}(P_{UX})$ and  $\phi \in \cB^{\eps_{11}}(P_{UY})$ be such that 
\begin{align}
\label{consm786}
H^{\eps_{11}}_{0}&(X|U)= \log \max_{u \in \cal{U}} |\mbox{Supp}(Q(X|U=u))|
\end{align}
and
\begin{align}
\label{smcon12}
D^{\eps_{11}}_{\infty}&(P_{UY}||P_U\times P_Y) \nonumber\\
& = \log \max_{(u,y) : P_{UY}(u,y)>0}\log\frac{\phi(u,y)}{P_{U}(u)P_{Y}(y)},
\end{align}
where
\beq
\label{phidef2}
\phi(U=u|Y=y) :=
\begin{cases}
\frac{\phi(u,y)}{P_{Y}(y)} &\mbox{if } P_{Y}(y)>0,\\ 
0        & \mbox{otherwise}.
\end{cases}
\enq
Notice that the triplet $(X,Y,U)$ satisfy the following
\beq
\label{ref}
X \rightarrow Y \rightarrow U.
\enq
For more details on \eqref{ref} see \cite[equation 4.4]{wyner-1975}. For every $(u,y) \in \mathcal{U}\times \cY$, let $g$ be a mapping such that
\beq
\label{fundef}
g(u,y):= \sum_{x\in\cX}P_{X|Y}(x|y){\bf{I}}(x,u),
\enq
where ${\bf{I}}(x,u)$ for every $(x,u) \in \cX \times  \mathcal{U}$ is defined as follows
\beq
{\bf{I}}(x,u) =
\begin{cases}
1 & \mbox{if } (x,u) \notin \mbox{Supp}(Q) , \\
0        & \mbox{otherwise.}
\end{cases}
\enq
Define the following set
\beq
\label{consset2345}
\mathcal{F} := \left\{(u,y)\in \mathcal{U} \times \cY : g(u,y) \leq \eps_{11}^{\frac{1}{2}} \right\}.
\enq

{{\bf{Random code generation:}}  Randomly and independently assign  an index $i \in [1: 2^{\ell_{\denc}^{\eps}(X)}]$ to every realization $x \in \cX$. The realizations with the same index $i$ form a bin $\cB (i)$. Randomly and independently generate $2^{\ell_{\denc}^{\eps}(Y)}$ realizations $u(k)$, $k \in [1:2^{\ell_{\denc}^{\eps}(Y)}]$, each according to $P_U$.}

{\bf{Encoding:}} If the encoder $1$ observes a realization $x \in \cB (i)$, then the encoder $1$ transmits $i$. For every realization $y \in \cY$ the encoder $2$ finds an index $k$ such that $(u(k),y) \in \mathcal{F}$. For the case when there are more than one such index, it sends the smallest one among them. If there is none, it then sends $k =1$.

{\bf{Decoding:}} The receiver finds the unique $x^{\prime} \in \cB(i)$ such that $(x^{\prime},u(k)) \in \mbox{Supp}(Q)$.

{\bf{Probability of error analysis:}} Let $M_1$ and $M_2$ be the chosen indices for encoding $X$ and $Y$. The error in the above mentioned encoding decoding strategy occurs if and only if one or more of the following error events occur
\begin{align*}
E_{1} &= \left\{(U(m_2),Y) \notin \mathcal{F},~ \forall m_2 \in \left[1:2^{\ell_{\denc}^{\eps}(Y)}\right]\right\},\\
E_2 &= \left\{(X, U(M_2)) \notin \mbox{Supp}(Q) \right\},\\
E_3 &=  \left\{\exists x^\prime \in \cB(m_1): (x^\prime, U(M_2)) \in \mbox{Supp}(Q), x^\prime \neq X\right\}.
\end{align*}
For more details on error events see \cite[Lemma 4]{Kuzuoka-2012}. The probability of error is upper bounded as follows
\beq
\label{errana123}
\Pr\{E\} \leq \Pr\{E_1\} + \Pr\{E^c_1 \cap E_2\} + \Pr\{E_3|X\in\cB(1)\}. 
\enq
We now calculate $\Pr\{E_1\}$ as follows
\begin{align}
&\Pr\{E_1\} \nonumber\\
&= \sum_{y \in \cY}P_{Y}(y)\Pr\bigg\{(U(m_2),y) \notin \mathcal{F}, \forall m_2 \in \left[1:2^{\ell_{\denc}^{\eps}(Y)}\right]\bigg\}\nonumber\\
&\overset{a}= \sum_{y \in \cY}P_{Y}(y) \left(1-\sum_{u:(u,y) \in \mathcal{F}}P_{U}(u)\right)^{2^{\ell_{\denc}^{\eps}(Y)}}\nonumber\\
&\overset{b} \leq\sum_{y \in \cY}P_{Y}(y) \bigg(1-2^{-D^{\eps_{11}}_{\infty}(P_{UY}||P_{U}\times P_{Y})}\nonumber\\
&\hspace{20mm}\sum_{u:(u,y) \in \mathcal{F}}\phi(U=u|Y=y)\bigg)^{2^{\ell_{\denc}^{\eps}(Y)}}\nonumber\\
& \overset{c} \leq 1-\sum_{(u,y) \in \mathcal{F}} \phi(u,y) + e^{-2^{\ell_{\denc}^{\eps}(Y)}2^{-D^{\eps_{11}}_{\infty}(P_{UY}||P_U\times P_Y)}} \nonumber \\
%&\overset{d} \leq \Pr\{(U,Y) \notin \mathcal{F}\} + e^{-2^{\ell_{\denc}^{\eps}(Y)}2^{-D^{\eps_{11}}_{\infty}(P_{UY}||P_U\times P_Y)}}\nonumber\\
& \overset{d} \leq \eps_{11} + \Pr\{(U,Y) \notin \mathcal{F}\} + e^{-2^{\ell_{\denc}^{\eps}(Y)}2^{-D^{\eps_{11}}_{\infty}(P_{UY}||P_U\times P_Y)}}, \nonumber
\end{align}
where $a$ follows because $U(1),\dots,U(2^{\ell_{\denc}^{\eps}(Y)})$ are independent and subject to identical distribution $P_U$; $b$ follows from \eqref{equalsupp}, \eqref{smcon12} and \eqref{phidef2}; $c$ follows from the inequalities $(1-x)^y \leq e^{-xy} ~ (0 \leq x \leq 1, y\geq 0)$ and $e^{-xy} \leq 1-y+x ~(x\geq 0, 0\leq y \leq 1)$ and \eqref{phidef2}; $d$ is true because of the following arguments
 \begin{align}
 1-\eps_{11} & \overset {a}\leq \sum_{(u,y) \in \mathcal{U} \times \cY} \phi(u,y) \nonumber\\
 & = \sum_{(u,y) \in \mathcal{F}^c} \phi(u,y)+  \sum_{(u,y) \in \mathcal{F}}\phi(u,y)\nonumber\\
 & \overset{b} \leq \Pr\{\mathcal{F}^c\} +  \sum_{(u,y) \in \mathcal{F}}\phi(u,y)\nonumber\\
 \label{rearrange4}
 &  \leq \Pr\{(U,Y) \notin \mathcal{F}\} +\sum_{(u,y) \in \mathcal{F}}\phi(u,y),
 \end{align}
 where $a$ and $b$ both follow from the fact that $\phi(u,y)
  \in \cB^{\eps_{11}}(P_{UY})$. By rearranging the terms in \eqref{rearrange4} we get
 \beq
 \label{errb12}
 1- \sum_{(u,y) \in \mathcal{F}}\phi(u,y) \leq \eps_{11}+\Pr\{(U,Y) \notin \mathcal{F}\}. 
 \enq
We now calculate $\Pr\{(U,Y)\notin\mathcal{F}\}$ as follows 
\begin{align}
\Pr\{(U,Y) \notin \mathcal{F}\} &= \Pr\{g(U,Y) \geq \eps_{11}^{\frac{1}{2}}\} \nonumber\\
&\overset{a} \leq \eps^{-\frac{1}{2}}_{11} \mathbb{E}_{UY}(g(U,Y)) \nonumber\\
& \leq \eps_{11}^{-\frac{1}{2}} \sum_{(u,y) \in \mathcal{U }\times \cY}P_{UY}(u,y) g(u,y) \nonumber\\
&\overset{b} = \eps_{11}^{-\frac{1}{2}} \sum_{(u,y) \in \mathcal{U}\times \cY}P_{UY}(u,y) \nonumber\\
& \hspace{10mm}\sum_{x\in\cX}P_{X|Y}(x|y){\bf{I}}(x,u) \nonumber\\
& = \eps^{-\frac{1}{2}}_{11}\sum_{\substack{(x,u,y) \in \cX\times \mathcal{U }\times \cY\\ (u,x) \notin \mbox{Supp}(Q)}}P_{XUY}(x,u,y) \nonumber\\
& = \eps^{-\frac{1}{2}}_{11} \sum_{(u,x) \notin \mbox{Supp}(Q)}P_{UX}(u,x) \nonumber\\
\label{finc3}
&\overset{c} \leq \eps^{\frac{1}{2}}_{11},
\end{align}
where $a$ follows from Markov's inequality; $b$ follows from \eqref{fundef} and $c$ follows because of the following arguments
\begin{align}
1-\eps_{11} & \overset {a}\leq \sum_{(u,x)\in \mathcal{U} \times \cX}Q(u,x)\nonumber\\
\label{rearrang5}
&\overset{b}\leq \sum_{(u,x) \in \mbox{Supp}(Q)}P_{UX}(u,x),
\end{align}
where $a$ follows from \eqref{consm786} and $b$ follows because $Q(u,x)\leq P_{UX}(u,x)$, for every $(u,x) \in \mathcal{U} \times \cX$. By rearranging the terms in \eqref{rearrang5} we get
\beq
1-\sum_{(u,x) \in \mbox{Supp}(Q)}P_{UX}(u,x) \leq \eps_{11}. \nonumber
\enq
The second term in \eqref{errana123} is calculated as follows
\begin{align}
\Pr&\{E^c_1\cap E_2\} \nonumber\\
&= \sum_{\substack{(x,u,y)\in \cX\times\cY\times\mathcal{U}\\(u,y) \in \mathcal{F}, (u,x)\notin \mbox{Supp}(Q)}}P_{XUY}(x,u,y) \nonumber\\
&=\sum_{(u,y) \in \mathcal{F}}P_{UY}(u,y)\sum_{x: (x,u) \notin  \mbox{Supp}(Q)}P_{X|Y}(x|y) \nonumber\\
\label{finc2}
&\leq \eps_{11}^{\frac{1}{2}},
\end{align}
where the last inequality follows from \eqref{consset2345}. From \eqref{errb12}, \eqref{finc3} and \eqref{finc2}  it follows that 
\begin{align}
%\label{crat}
\Pr\{&E_1\}+\Pr\{E^c_1 \cap E_2\} \nonumber\\
&\leq \eps_{11} + 2\eps_{11}^{\frac{1}{2}}+ e^{-2^{\ell_{\denc}^{\eps}(Y)}2^{-D^{\eps_1}_{\infty}(P_{UY}||P_U\times P_Y)}}. \nonumber
\end{align}
Let 
\begin{align}
\label{finc4}
\eps_1 &\geq \eps_{11} + 2\eps_{11}^{\frac{1}{2}}+e^{-2^{\ell_{\denc}^{\eps}(Y)}2^{-D^{\eps_1}_{\infty}(P_{UY}||P_U\times P_Y)}}.
\end{align}
It now easily follows that
\begin{align}
\ell_{\denc}^{\eps}(Y) & \geq \log[-\ln (\eps_1-\eps_{11}-2\eps_{11}^{\frac{1}{2}}) ] \nonumber\\
&\hspace{5mm}+D^{\eps_{11}}_{\infty}(P_{UY}||P_U\times P_Y). \nonumber
\end{align}
Finally, the third term in \eqref{errana123} is calculated as follows 
\begin{align}
\Pr&\{E_3\} = \Pr\{E_3|\cB(1)\}\nonumber\\
& = \sum_{(x,u) \in \cX \times \mathcal{U}} \Pr\big\{(X,U) = (x,u)|X \in \cB(1)\big\}\nonumber\\ 
 &\hspace{10mm} \Pr \bigg\{\exists x^\prime \neq x : x^\prime \in \cB(1)\mbox{and} ~Q(x^\prime,u)>0\nonumber \\
 &\hspace{18mm}\big| x \in \cB(1),(X,U) = (x,u)\bigg\} \nonumber\\
& \leq \sum_{(x,u) \in \cX \times \mathcal{U}}P_{XU}(x,u) \sum_{\substack {x^\prime \neq x \\ Q(x^\prime,u) > 0}}\Pr \{x^\prime \in \cB(1)\} \nonumber\\
&\leq  2^{-\ell_{\denc}^{\eps}(X)}\sum_{(x,u) \in \cX \times \mathcal{U}}P_{XU}(x,u)  \max_{u \in \mathcal{U}} \sum_{x:Q(x,u)>0} 1 \nonumber\\
& \overset{a} = 2^{-\ell_{\denc}^{\eps}(X)}\nonumber\\
\label{sw123}
&\hspace{5mm}\sum_{(x,u) \in \cX \times \mathcal{U}}P_{XU}(x,u) \max_{u \in \mathcal{U}} |\mbox{Supp}(Q(X|U=u))|\\
& = 2^{-\ell_{\denc}^{\eps}(X)}\max_{u \in \mathcal{U}} |\mbox{Supp}(Q(X|U=u))|, \nonumber
\end{align}
where $a$ follows because $Q(X=x|U=u): = \frac{Q(x,u)}{P_U(u)}$ and $Q(X=x|U=u): =0$ if $P_{U}(u) = 0$.
Thus from \eqref{finc4} and \eqref{sw123} it follows that
\beq
\Pr\{E\} \leq \eps_1 + 2^{-\ell_{\denc}^{\eps}(X)}\max_{u \in \mathcal{U}} |\mbox{Supp}(Q(X|U=u))|. \nonumber
\enq
Let
\beq
\eps_1 + 2^{-\ell_{\denc}^{\eps}(X)}\max_{u \in \mathcal{U}} |\mbox{Supp}(Q(X|U=u))| \leq \eps. \nonumber
\enq
It now easily follows that 
\beq
\ell_{\denc}^{\eps}(X) \geq H^{\eps_{11}}_{0}(X|U) - \log (\eps-\eps_1). \nonumber
\enq
%\end{enumerate}
This completes the proof.
\end{proof}

The asymptotic optimality of the rate region obtained in Theorem \ref{oneshotsideinf} is an immediate consequence of Definition \ref{asymsiinf}, Lemma \ref {assymcond} and \ref{asyminfinityrate}.

\section{Conclusion and Acknowledgements}
We proved that smooth max divergence and smooth max conditional \renyi entropy can be used to obtain one-shot achievable rate region for source coding when coded side information is available at the decoder. Furthermore, we showed that asymptotically this region coincides with the rate region as derived by Wyner in \cite{wyner-1975}.

The author gratefully acknowledges the helpful discussions with Mohit Garg and Sarat Moka.

\bibliographystyle{ieeetr}
\bibliography{master}

\end{document}